\documentclass[a4paper]{article}
\usepackage[utf8]{inputenc}

\usepackage[margin=1.3in]{geometry}

\usepackage{amsmath,amssymb,amscd,amsthm,amsfonts}
\usepackage{graphicx,subcaption}
\usepackage{hyperref}
\usepackage{dsfont}
\usepackage{xcolor}
\usepackage{color,soul}
\usepackage[utf8]{inputenc}
\usepackage[capitalise]{cleveref}
\usepackage{thm-restate}

\newtheorem{theorem}{Theorem}[section]

\newtheorem{lemma}[theorem]{Lemma}

\newtheorem{question}[theorem]{Question}

\theoremstyle{definition}

\newtheorem{definition}[theorem]{Definition}

\graphicspath{{figures/}}

\author{Mats Bierwirth \\
        Dept. of Computer Science, ETH Z\"{u}rich, Switzerland \\ {\tt mbierwith@student.ethz.ch}
        \and
        Julia H\"{u}tte \\
        Dept. of Computer Science, ETH Z\"{u}rich, Switzerland \\ {\tt jhuette@student.ethz.ch}
        \and
        Patrick Schnider \\
        Dept. of Mathematics and Computer Science,
        University of Basel\\
        Dept. of Computer Science, ETH Z\"{u}rich, Switzerland \\ {\tt patrick.schnider@inf.ethz.ch}
        \and
        Bettina Speckmann \\
        Dept. of Mathematics and Computer Science, TU Eindhoven, The Netherlands \\ {\tt b.speckmann@tue.nl}
}
\title{Bounds for $k$-centers of point sets\\ under $L_{\infty}$-bottleneck distance}
\date{}

\begin{document}

\maketitle

\begin{abstract}
We consider the $k$-center problem on the space of fixed-size point sets in the plane under the $L_{\infty}$-bottleneck distance. While this problem is motivated by persistence diagrams in topological data analysis, we illustrate it as a \emph{Restaurant Supply Problem}: given $n$ restaurant chains of $m$ stores each, we want to place supermarket chains, also of $m$ stores each, such that each restaurant chain can select one supermarket chain to supply all its stores, ensuring that each store is matched to a nearby supermarket.
How many supermarket chains are required to supply all restaurants?
We address this questions under the constraint that any two restaurant chains are close enough under the $L_{\infty}$-distance to be satisfied by a single supermarket chain. We provide both upper and lower bounds for this problem and investigate its computational complexity.
\end{abstract}

\section{Introduction}
The \emph{$k$-center problem} is a classical problem in computational geometry \cite{10.1145/509907.509947, 10.1145/1376916.1376944, 10.5555/3001460.3001507, Handler1979LocationON, doi:10.1137/0213014}. Given a set $P$ of points in some metric space $X$, the goal is to partition them into $k$ parts $P_1,\ldots,P_k$ such that for each part $P_i$ there is some other point $x_i\in X$ that is close to each point $p\in P_i$. While this problem has mainly been studied for points in Euclidean space, motivated by \emph{persistence diagrams} in topological data analysis we consider $X$ as the space of $m$ unordered points in $\mathbb{R}^2$ under the $L_{\infty}$-bottleneck distance. Intuitively, the bottleneck distance is defined as the length (in $L_{\infty}$-distance) of the longest edge in a perfect matching minimizing said length. Computing such a bottleneck distance is another classical problem in computational geometry \cite{Burkard1980, efrat1996improvements, KARIMABUAFFASH2014447, KATZ2023101986}.
The main difference of our setting to persistence diagrams is that the latter contain infinitely many points on the so-called diagonal. Nevertheless, we hope that some of our ideas might be used to construct $k$-centers for persistence diagrams, at least under some additional assumptions. Such $k$-centers could be interesting for example for clustering persistence diagrams, a topic that has recently gained attention \cite{cao2024k, cohen2005stability, kerber2017geometry, lacombe2018large, marchese2017k, 8794517}.

In order to illustrate the problem, we give another interpretation in terms of supplying restaurant chains with supermarket chains, which we call the \emph{Restaurant Supply Problem}:
In a city, for example Manhattan, there are $n$ different restaurant chains that have $m$ stores each. Their supply is secured by $k$ supermarket chains that also have $m$ stores each. All restaurants that belong to the same chain get their supply from the same supermarket chain. However, each restaurant within one chain gets it from a different store. This means that as soon as a supermarket chain gets chosen by a restaurant chain, each store gets matched to one specific restaurant that it supplies. We are interested in the number of supermarket chains needed to satisfy all restaurant chains. Formally, we define the following:

\begin{definition}
Let $R_1=(r_1^1, r_1^2,\ldots, r_1^m),\ldots,R_n=(r_n^1,\ldots r_n^m)$ be $n$ restaurant chains with $m$ stores each in some metric space $X$. We say that a supermarket chain consisting of $m$ supermarkets $s_1,\ldots,s_m$ \emph{$\delta$-satisfies} (or equivalently just \emph{satisfies}) a set of restaurant chains $R_a,\ldots, R_b$ if there exists a relabeling of the stores in each $R_i$ such that the distance between $s_j$ and $r_i^j$ is at most $\delta$ for each $i$ and $j$. More generally, we say that $k$ supermarket chains $\delta$-satisfy the $n$ restaurant chains $R_1\ldots,R_n$ if $R_1,\ldots, R_n$ can be partitioned into $k$ subfamilies, each of which can be $\delta$-satisfied by a single supermarket chain.
\end{definition}

Of course, if the restaurant chains are operating sufficiently far from each other, then any supermarket chain can only satisfy one restaurant chain, so in general there are instances where we need $n$ supermarket chains. For this reason, we focus on the case where the restaurant chains are in actual competition and operate close to each other. More specifically, we will assume that any $h$ restaurant chains can be satisfied by a single supermarket chain. This assumption can make a big difference: consider the case where each restaurant chain only has a single store, and assume that distances are measured in Euclidean distance. Then a single supermarket satisfies some $h$ restaurants whenever it lies in the intersection of the disks with radius $\delta$ centered at the locations of the restaurants. As these disks are convex, it follows from Helly's theorem that if any 3 of them have a common intersection, then all of them do. In other words, if any 3 restaurants can be satisfied by a single supermarket, then all restaurants can be satisfied by a single supermarket.

The Restaurant Supply Problem is now the optimization problem of minimizing the number of supermarket chains under this assumption.
In the rest of this manuscript, we will restrict our attention to cities that, like Manhattan or large parts of Barcelona, are laid out on a square grid and thus distances are measured differently\footnote{Formally speaking, in such cities distances are measured in the $L_1$-metric. However, as we only care about metric balls in our arguments, and as metric balls are squares in the plane under both metrics, in order to be consistent with our motivation from persistence diagrams, we choose to work with the $L_{\infty}$-metric.}. More formally, we define the distance between two points in the plane as the $L_{\infty}$-distance. Then the disks of radius $\delta$ are axis-aligned squares and from Helly's theorem for boxes we get an even stronger statement for chains with only one store: if any two restaurants can be satisfied by a single supermarket, then all of them can. For this reason, we set $h=2$ for the rest of this manuscript, that is, we assume that any two restaurant chains can be satisfied by a single supermarket chain.

\section{An upper bound}


\begin{theorem}\label{thm:upper_bound}
    Let $\{r^1_1,\ldots,r^m_1\},\ldots,\{r^1_n,\ldots,r^m_n\}$ be $n$ restaurant chains with $m$ stores each, such that any two of them can be satisfied by a single supermarket chain. Then all of them can be satisfied by $k=m!$ supermarket chains.
\end{theorem}

    For each restaurant chain, we consider its $m$ stores and annotate them with the numbers $1,\ldots,m$ from left to right. Now consider the permutation $\pi$ obtained when enumerating the points of that restaurant chain from bottom to top. If any two points have the same $x$- or $y$-coordinate, we break ties lexicographically.
    We claim that all restaurant chains from the same permutation can be satisfied by just one supermarket chain. 
    \begin{lemma} \label{lem:UpperBound}
        For two restaurants $r_1, r_2$ annotated with the same index from chains $R_1, R_2$ in the same permutation, their distance is at most $2\delta$.
    \end{lemma}
    \begin{proof}
    The restaurants $r_1, r_2$ split the plane into four quadrants each, see Figure \ref{fig:upper_bound}. Since both stores are annotated with the same index and belong to chains from the same permutation, their quadrants contain the same number of points. Let $a$ be the number of points in the top-left, $b$ in the top-right, $c$ in the bottom-left and $d$ in the bottom-right quadrant. Furthermore, let $s_1$ and $s_2$ be the $\delta$-balls (which are squares in $L_{\infty}$) around $r_1$ and $r_2$. For sake of contradiction, assume that they do not intersect. By the separation theorem, there exists either a horizontal or vertical separation line between them. Without loss of generality consider the case of a horizontal separation line $\ell$ and let $r_2$ be above $r_1$.
    
    \begin{figure}[t]
        \centering
        \includegraphics[]{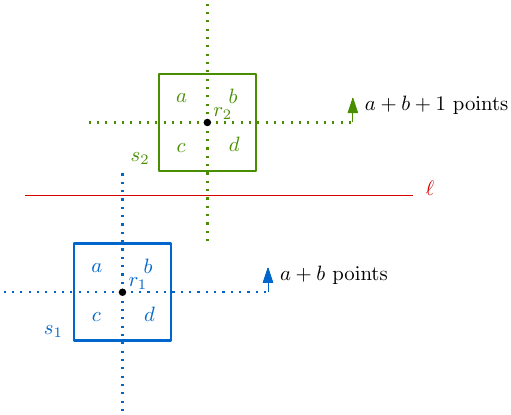}
        \caption{Two restaurants annotated with the same index from chains in the same permutation split the plane into quadrants with the same number of restaurants.}
        \label{fig:upper_bound}
    \end{figure}
    
    \noindent By the Helly-like assumption of our theorem, $R_1, R_2$ can be satisfied using only one supermarket chain. Thus, their stores can be matched in a way such that the distance of the longest matching edge is at most $2\delta$.
    There are at least $a+b+1$ restaurants in $R_2$ (the upper two quadrants plus $r_2$ itself) that can in such a matching only be matched to restaurants strictly above $\ell$. However, $R_1$ can contain at most $a+b$ such restaurants. By contradiction, $s_1$ and $s_2$ intersect.
    \end{proof}
    
\begin{proof}[Proof of Theorem \ref{thm:upper_bound}]
    By Lemma \ref{lem:UpperBound} any two stores  annotated with the same index from chains in the same permutation are at most distance $2\delta$ apart. In other words, their $\delta$-balls intersect, and thus, the statement follows from Helly's theorem for boxes.
\end{proof}

\section{A lower bound}

Unfortunately, the bound proven in the previous section has a superexponential dependence on the number of stores per chain $m$. As it turns out, no subexponential bound exists.

\begin{restatable}{theorem}{lowerbound}\label{thm:lower_bound}
    There exists a placement of $n$ restaurant chains $\{r^1_{1}, \dots, r^m_{1}\},\ldots,\{r^1_{n}, \dots, r^m_{n}\}$, with $m$ stores each, such that any two can be satisfied by a single supermarket chain, but satisfying all $n$ chains requires  $k \geq \min(n,\lfloor e^{m/(27+\varepsilon)}\rfloor) /2$ supermarket chains, for any constant $\varepsilon>0$.
\end{restatable}

The proof is based on the following construction, an example where a single supermarket chain does not suffice. Consider six points arranged on a circle such that any two squares of radius $\delta$ centered at the points intersect if and only if the points are neighbors on the circle. We refer to this arrangement as a city; an example is illustrated in Figure \ref{fig:city}. 

\begin{figure}[t]
    \centering
    \includegraphics[]{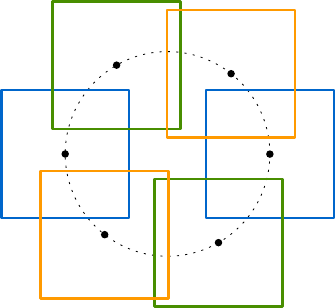}
    \caption{A city with 3 restaurant chains, each consisting of 2 stores placed on antipodal points.}
    \label{fig:city}
\end{figure}

\begin{proof} Since any two restaurant chains can be satisfied by a single supermarket chain, the minimum number of supermarket chains required to satisfy $n$ restaurant chains is at most $n/2$. Thus, assume $n \geq e^{m/(27+\varepsilon)}$. Let $n_{\max} = \lfloor e^{m/(27+\varepsilon)}\rfloor$. We construct a placement for the first $n_{\max}$ restaurant chains that requires at least $n_{\max} /2$ supermarket chains to satisfy. This then entails the theorem.

Consider $m/2$ cities, sufficiently far away from each other to not interfere with one another. In each city, each restaurant chain constructs independently and uniformly at random two stores on opposite points. We show that for any three restaurant chains, with high probability there exists a city where these three restaurant chains can not be satisfied by a single supermarket chain. This then entails, that $n_{\max}/2$ supermarket chains are needed to satisfy all restaurant chains.

For a given city and three fixed restaurant chains there are twenty-seven possible store placements of these three chains within this city. For six of those, one supermarket chain is not sufficient to satisfy the three restaurant chains. Since each placement is equally likely, the probability of this happening is $2/9$. Using independence of the cities and the inequality $1-x \leq e^{-x}$, the probability that all three chains can be satisfied in all cities is bounded by
\[
    \mathbb{P}[\text{Three restaurant chains satisfied}] = (1-2/9)^{m/2} \leq e^{-m/9}
    \leq \lfloor e^{m/9} \rfloor ^{-1} = n_{\max}^{-3-\varepsilon'}.
\]
for some constant $\varepsilon'>0$ dependent on $\varepsilon$. Applying a union bound over all $O(n_{\max}^3)$ restaurant chain triples, the probability that any triple can be satisfied by a single supermarket chain is $n_{\max}^{-\varepsilon'}$ which is infinitesimal in $n_{\max}$. Thus, with high probability, this procedure produces a placement where no triple can be satisfied by a single supermarket chain.
\end{proof}

\section{Computational complexity}

In this section we analyze the computational complexity of deciding whether all restaurant chains can be satisfied by a single supermarket chain. 
In both our results it will be helpful to take a graph theoretic viewpoint on the problem: we define the \emph{restaurant graph} $G$ whose vertex set is the set of all restaurant stores, and where two stores $r^b_a$ and $r^y_x$ of different chains are connected whenever their $\delta$-balls intersect. The vertex set of this graph partitions naturally into $m$-sets of vertices that correspond to the $m$ stores of a restaurant chain. We interpret the different chains as colors and say that a clique $C$ in $G$ is a \emph{colorful clique} if it contains exactly one vertex of each $m$-set associated to a restaurant chain. 

\begin{lemma}\label{lem:restaurant_graph}
    Let $\{r^1_1,\ldots,r^m_1\},\ldots,\{r^1_n,\ldots,r^m_n\}$ be $n$ restaurant chains with $m$ stores each. Then all of them can be satisfied by a single supermarket chain if and only if the vertex set of the restaurant graph $G$ can be partitioned into $m$ colorful cliques.
\end{lemma}

\begin{proof}
    Assume there is a placement of the $m$ supermarket stores $s_1,\ldots,s_m$ such that each restaurant $r^i_j$ is satisfied by the store $s_i$. In particular, any two vertices $r^i_a$ and $r^i_b$ are connected and we can thus partition the vertex set of $G$ into $m$ colorful cliques. On the other hand if such a partition into colorful cliques exists, then for any two stores in this clique their $\delta$-balls intersect. Thus, by Helly's theorem for boxes, all of them intersect, which means that we can place a supermarket store in this intersection to satisfy all the stores in the clique.
\end{proof}

The following can be proven using standard methods. 

\begin{restatable}{theorem}{algorithm}\label{thm:algorithm}
    Let $\{r^1_1,r^2_1\},\ldots,\{r^1_n,r^2_n\}$ be $n$ restaurant chains, each having two stores, such that any two of them can be satisfied by a single supermarket chain. Then we can decide in time $O(n^3)$ whether all of them can be satisfied with a single supermarket chain.
\end{restatable}

\begin{proof}
    First, build the restaurant graph $G$. By Lemma \ref{lem:restaurant_graph}, we want to decide whether the vertex set of $G$ can be partitioned into two colorful cliques. Consider now the subgraph induced by the four stores of two restaurant chains $A$ and $B$. This is a subgraph of the bipartite graph $K_{2,2}$. If it is not the complete bipartite graph $K_{2,2}$ then it contains at most one perfect matching. In this case, we say that the pair $(A,B)$ of vertex pairs is \emph{forced}. If the subgraph does not contain any perfect matching, then we call the pair $(A,B)$ \emph{incompatible}. Note that as any two restaurant chains can be satisfied by a single supermarket chain, $G$ does not contain any incompatible pairs (yet).

    Given a forced pair $(A,B)$ on the vertices $a_1,a_2\in A$ and $b_1,b_2\in B$, we perform the following \emph{contraction}: assume without loss of generality that the unique matching between $A$ and $B$ connects $a_1$ with $b_1$ and $a_2$ with $b_2$. Create a new vertex pair $C=(c_1,c_2)$ and remove $A$ and $B$ from the graph $G$. Let $v$ be any other vertex in $G$. Connect $v$ to $c_i$ if and only if $v$ was connected to both $a_i$ and $b_i$. Note that in this contraction step we might create an incompatible pair involving $C$.

    Our algorithm now proceeds as follows: as long as there are forced pairs in $G$, contract them. As soon as there is an incompatible pair, return that there is no solution. If we end up in the situation that no pair is forced or incompatible, return that there is a solution.

    Clearly, this algorithm runs in time $O(n^3)$: building the graph takes time $O(n^2)$, each contraction takes time $O(n)$, and we perform at most $O(n^2)$ many contractions.

    It follows from the construction of the contraction step that the original graph can be partitioned into colorful cliques if and only if the contracted graph can, which proves the correctness of our algorithm.
\end{proof}

Let us mention here that we did not try to optimize the runtime and it is thus likely that it can still be improved. Further, it is an interesting problem whether for a larger constant number of stores we can still solve the problem in polynomial time.

\begin{question}
    Given $n$ restaurant chains, each having $m$ stores, where $m$ is a constant, such that any two of them can be satisfied by a single supermarket chain, is there a polynomial time algorithm to decide whether all of them can be satisfied by a single supermarket chain?
\end{question}

On the other hand, if the number of stores is unbounded, then the problem becomes NP-complete, even if there are only 3 restaurant chains.

\begin{restatable}{theorem}{hardness}\label{thm:hardness}
    Let $\{r^1_1,\ldots,r^m_1\},\{r^1_2,\ldots,r^m_2\},\{r^1_3,\ldots,r^m_3\}$ be 3 restaurant chains with $m$ stores each, such that any two of them can be satisfied by a single supermarket chain. Then it is NP-complete to decide whether all of them can be satisfied with a single supermarket chain.
\end{restatable}


\begin{figure}[t]
        \centering
        \includegraphics[]{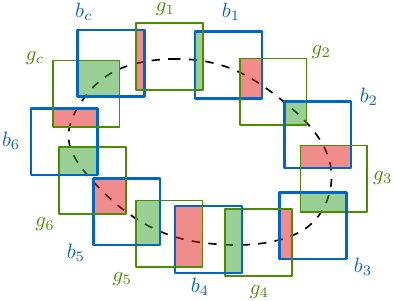}
        \caption{Placing supermarkets in the green (red) regions sets the variable to ``true'' (``false'').}
        \label{fig:variable}
    \end{figure}

\begin{figure}[b]
        \centering
        \includegraphics[]{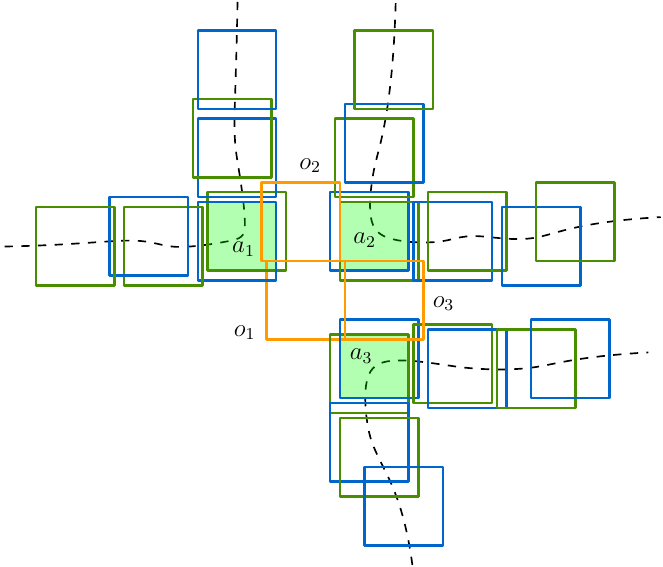}
        \caption{A clause gadget.}
        \label{fig:clause}
    \end{figure}

    \begin{figure}[t]
        \centering
        \includegraphics[]{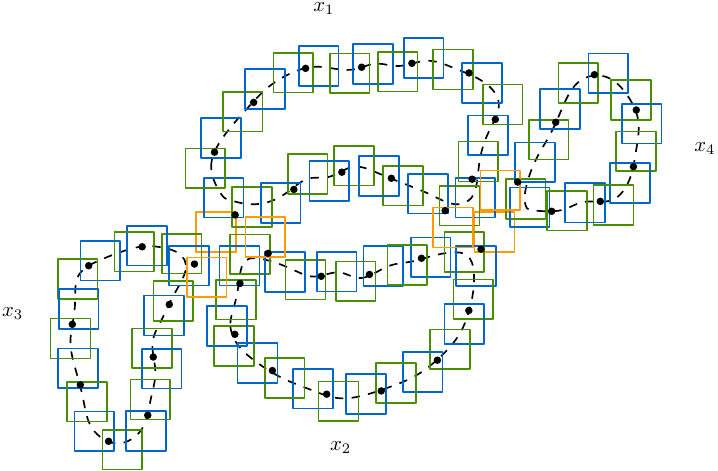}
        \caption{All gadgets for the (very small) formula $\varphi=(x_1\vee\neg x_2\vee x_3)\wedge(\neg x_1\vee x_2\vee x_4)$. The black points describe a placement of supermarkets corresponding to setting all variables to ``true''.}
        \label{fig:global}
    \end{figure}

\begin{proof}
    The containment in NP follows from Lemma \ref{lem:restaurant_graph}.
    In order to show that the problem is NP-hard, we use a reduction from planar 3-SAT \cite{planar3sat}. We describe our construction in terms of the $\delta$-balls. For each of the three restaurant chains we give a different color, namely blue, green and orange. Thus, for our reduction, given a planar SAT formula $\varphi$ with $m$ clauses we want to place $n$ blue squares, $n$ green squares and $n$ orange squares such that we can partition them into colorful triples with common intersections if and only if $\varphi$ is satisfiable. We now describe such a placement.

    \smallskip\noindent{\bf Variable gadgets.}
    For an illustration of the variable gadget, see Figure \ref{fig:variable}.
    For each variable in the formula $\varphi$ we use $c$ green and blue squares $g_1,\ldots,g_c$ and $b_1,\ldots,b_c$, where $c$ is some large enough number that is polynomial in $m$ and will be determined later. We arrange the squares in a cycle $g_1,b_1,g_2,b_2\ldots,g_c,b_c$ such that $g_i$ intersects exactly $b_{i-1}$ and $b_i$ and $b_i$ intersects exactly $g_i$ and $g_{i+1}$ (we set $b_0=b_c$ and $g_{c+1}=g_1$). Note that there are exactly two types of placements of supermarkets to satisfy the involved restaurants: we either place them in the intersection $g_i\cap b_i$ or in the intersections $b_i\cap g_{i+1}$. We will interpret the former as setting the variable to ``true'' and the latter as setting it to ``false''.

    \medskip\noindent{\bf Clause gadgets.}
    For an illustration of the clause gadget, see Figure \ref{fig:clause}.
    Consider the clause $(\ell_1, \ell_2, \ell_3)$, where $\ell_i$ is a literal of the variable $x_i$, that is, it is either $x_i$ or $\neg x_i$. Denote by $a_i$ an intersection in the variable gadget of $x_i$ corresponding to a placement of supermarkets setting $\ell_i$ to ``true''. We now place an orange square $o_1$ in such a way that it intersects exactly all three $a_i$'s. Note that in this way, a supermarket satisfying $o_1$ needs to lie in one of the intersections $a_i$, meaning that the corresponding literal is true. We place two more orange squares $o_2$, intersecting $a_1$ and $a_2$ as well as the intersections on the variable gadgets before them, and $o_3$, intersecting $a_2$ and $a_3$ as well as the intersections after them, see Figure \ref{fig:clause}.

    \medskip\noindent{\bf Linking the gadgets.}
    For an illustration of this step, see Figure \ref{fig:global}.
    Consider the variable-clause graph $G_{\varphi}$ of the formula $\varphi$, i.e., the graph whose vertices are the variables and clauses of $\varphi$, with a connection between variable $x_i$ and clause $C_j$ if the literal $\ell_i$ appears in $C_j$. As $\varphi$ is a planar 3-SAT formula, this graph is planar. Fix an embedding of $G_{\varphi}$ on a polynomial grid, e.g.\ by the algorithm of de Fraysseix, Pach and Pollack \cite{de1990draw}. For each vertex corresponding to a variable, route the variable gadget along its edges and place a clause gadget at each clause vertex. As the embedding is on a polynomial grid, this can be done with polynomially many green and blue squares. By construction, the number of green and blue squares is equal. Further, we have so far used 3 orange squares for each clause.
    
    \medskip\noindent{\bf Adding the remaining orange squares.}
    For an illustration of the placement of the additional orange squares, see Figure \ref{fig:orange}.
    So far we have only added 3 orange squares for each clause, that is, one per variable-clause incidence. We first add one more orange square for each such incidence as follows: let $b_i$ be the blue square contributing to the intersection setting the literal to true and consider the 2 intersection before and the 2 intersections after it on the variable cycle. We place an orange square in such a way that it intersects all those 5 intersections (see Figure \ref{fig:orange}). We further remember $b_i$ and $b_{i-1}$ as \emph{marked}. Note that we have 2 marked blue squares and 2 orange squares per variable-clause incidence. Finally, we add an orange square at the location of each unmarked blue square. We thus have the same number of orange squares as blue squares and hence also the same number as green squares.

    \medskip\noindent{\bf Correctness.}
    We first show that we have constructed a valid instance of the problem, that is, that any two restaurant chains can be satisfied by a single restaurant chain. For the green and blue squares this follows from the construction: all the variable cycles are independent and each cycle can be satisfied in two ways. For the blue and orange squares, all except the marked squares are copies of each other, so it suffices to show that the marked squares can be satisfied. This follows from the placement of the orange squares at a clause gadget. A similar argument applies to the green and orange squares.
    
    By the construction of the clause gadgets and the placements of the orange squares it also follows that any satisfying assignment of the formula $\varphi$ defines a placement of supermarket stores that satisfy all three restaurant chains. Finally, any placement of supermarket stores that satisfy all three restaurant chains can be mapped back to an assignment which by construction satisfies $\varphi$.
\end{proof}

\begin{figure}[t]
        \centering
        \includegraphics[]{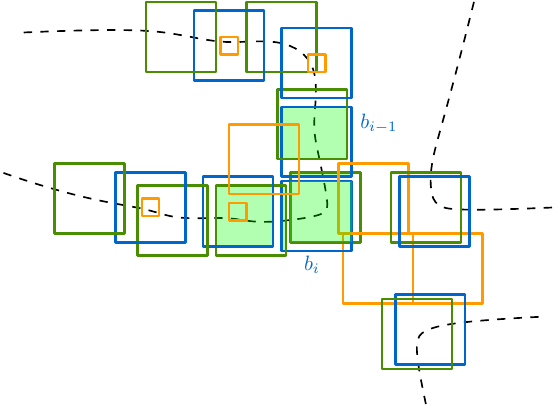}
        \caption{The placement of the remaining orange squares. The orange squares that are copies of blue squares are drawn small.}
        \label{fig:orange}
    \end{figure}

\medskip\noindent{\bf Acknowledgments.}
    This work was initiated during the 21st Gremo's Workshop on Open Problems in Pura, Switzerland. The authors would like to thank the other participants for the lively discussions. The first two authors were supported by the German Academic Scholarship Foundation. The second author was further supported by a fellowship of the German Academic Exchange Service (DAAD) and the Cusanuswerk.

\bibliography{refs}

\end{document}